\documentclass[journal]{IEEEtran_Kit}
\usepackage{amssymb}
\usepackage{mathtools}
\usepackage{diagbox}
\usepackage{graphicx,cite,epsfig,amssymb,amsmath}
\usepackage{pifont}
\usepackage{bm}
\usepackage{color}
\usepackage{stmaryrd}
\usepackage{standalone}
\usepackage{preview}
\usepackage{xcolor}
\usepackage{subfigure}
\usepackage{textcomp}
\usepackage{caption}
\usepackage{bbding}
\usepackage[ruled,vlined]{algorithm2e}
\usepackage{algorithmic}
\usepackage{booktabs}
\usepackage{epstopdf}
\usepackage{hyperref}
\hypersetup{hypertex=true,
  		    colorlinks=true,
            linkcolor=blue,
            anchorcolor=blue,
            citecolor=blue}
\usepackage[utf8]{inputenc}

\begin{document}
\title{Convolutional Polar Coded Modulation}

\author{Kangjian Qin, ~\IEEEmembership{Member,~IEEE, }
and Zhaoyang Zhang, ~\IEEEmembership{Member,~IEEE}
\thanks{K. Qin (e-mail: antonyqkj@163.com) is with the College of Information Science and Electronic Engineering, Zhejiang University, Hangzhou 310027, China, and is also with the Department of Electrical-Electronics Engineering, Bilkent University, Ankara, 06800, Turkey. Z. Zhang (Corresponding Author, e-mail: zhzy@zju.edu.cn) is with the College of Information Science and Electronic Engineering, Zhejiang University, Hangzhou 310027, China, and is also with the International Joint Innovation Center, Zhejiang University, Haining 314400, China.}
\thanks{This work was supported in part by the China Scholarship Council, National Natural Science Foundation of China under Grant 61725104, National Key R\&D Program of China under grant 2018YFB1801104.}
}

\maketitle

\begin{abstract}
$2^m$-ary modulation creates $m$ bit channels which are neither independent nor identical, and this causes problems when applying polar coding because polar codes are designed for independent identical channels.
Different from the existing multi-level coding (MLC) and bit-interleaved coded modulation (BICM) schemes, this paper provides a convolutional polar coded modulation (CPCM) method that preserves the low-complexity nature of BICM while offering improved spectral efficiency. Numerical results are given to show the good performance of the proposed method.
\end{abstract}


\section{Introduction}
\PARstart{C}{oded} modulation is of vital importance in spectral efficient communications. For better performance, polar code \cite{Channel_polarization_Arikan}, which provably achieves the symmetric capacity of binary-input discrete memoryless channels (BDMCs), has been considered with higher-order modulation in various works \cite{PCM_Huber, Compound_Polar_Mahdacifar_2013, BICM_Mahdavifar_2016, Arikan_2015}. In this paper, we aim to further improve the spectral efficiency of polar coded modulation.

Fig. \ref{flowchart} shows our system model. We consider transmission over discrete-time memoryless AWGN channel with real number output, thus one dimensional constellation, i.e., pulse-amplitude modulation (PAM), is used in this paper. With $2^m$-ary PAM, $m$ bit channels are created. These bit channels do not have the same quality, moreover, they are not independent. In applying polar coding, this cause problem because polar codes are designed for independent identical channels. Multi-level coding (MLC) circumvents this problem by splitting the $2^m$-ary channels into $m$ layers of bit channels and using multiple codes that are tailored to each layer \cite{New_MLC_Imai_1997,MLC_Huber_1999}, however, it suffers form complexity and latency. Bit-interleaved coded modulation (BICM) scheme ignores the dependencies among bit channels and uses one single unified code over $2^m$-ary channel, however, this is not rigorous \cite{BICM_Zehavi_1992,BICM_Caire_1998}. Following the idea of MLC and BICM, various polar coded modulation schemes were proposed in \cite{PCM_Huber, Compound_Polar_Mahdacifar_2013, BICM_Mahdavifar_2016}, and a summary of polar coded modulation was given in \cite{Arikan_2015}.

Here, we take a different approach to solve the above problem faced by polar coding. Take $4$-PAM for example, each symbol $s\in\{-3A,-A,A,3A\}$ can be labeled by two bits $\{v_1,v_2\}$. There are totally $4!$ labeling rules and we give an example of the set-partition labeling in Fig. \ref{SPL}. The bit channel experienced by $v_1$ is denoted as B, and the bit channel experienced by $v_2$ is denoted as G, and their capacities are denoted as $I(B)$ and $I(G)$, respectively. We plot $I(B)$ and $I(G)$ using both MLC and BICM principles in Fig. \ref{4PAM_Bit_level_capa}, where for MLC $I(B)=I(V_1;Y)$ and $I(G)=I(V_2;Y|V_1)$; and for BICM $I(B)=I(V_1;Y)$ and $I(G)=I(V_2;Y)$. The channel B is less reliable than G in the sense that $I(B)<I(G)$.
To apply polar coding, we first study the construction of polar codes under independent non-identical channels. As shown in Fig. \ref{Mapping}, when there are two types of independent bit channels, we found that interleaver-$2$ induces the lowest block error rate (BLER) union bound for a polar code, whereas interleaver-$1$ and other random interleavers in general give worse BLER performance. We also found that when there are $m$ ($m\geq 2$) types of channels, this conclusion holds generally for polar codes with medium to long block-length. Based on this, we construct polar codes by combining different bit channels at each polarization kernel as long as possible in this paper.

Also note that B and G are not independent as they come from the same symbol $s$. To deal with such dependency, a convolutional polar coded modulation (CPCM) method is proposed, and we give a toy example in Fig. \ref{cpcm4pam} for $4$-PAM: given a polar code $a_1^N=\{a_1,\dots,a_N\}$, we propose to map its even-indexed bits to the good bit channels in current transmission block, whereas its odd-indexed bits are mapped to the bad channels in next transmission block. In this manner, all the bit channels used by a polar code are independent. Moreover, all the involved polar codes share one single unified construction. This idea is then generalized to $2^m$-ary PAM cases with $m\geq 2$. Finally, we indicate that the proposed CPCM method asymptotically achieves the constellation capacity.

\section{System Model}\label{sec:2}
\begin{figure}[t]\centering
\includegraphics[width=3.3 in]{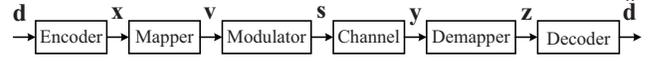}
\caption{\small The high level diagram of coded modulation.}
\label{flowchart}
\vspace{-0.3cm}
\end{figure}

\begin{figure}
\centering
\includegraphics[width=2.25 in]{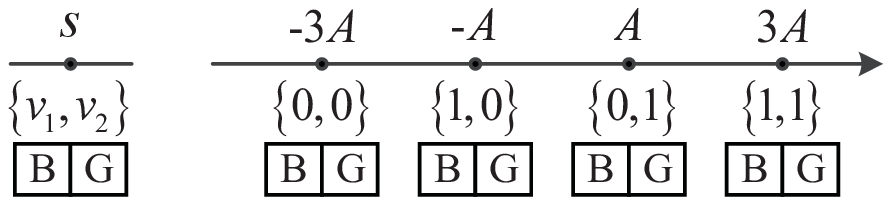}
\caption{\small Set-partition labeling for 4-PAM.}
\vspace{-0.5cm}
\label{SPL}
\end{figure}

\begin{figure}[t]\centering
\includegraphics[width=2.2 in]{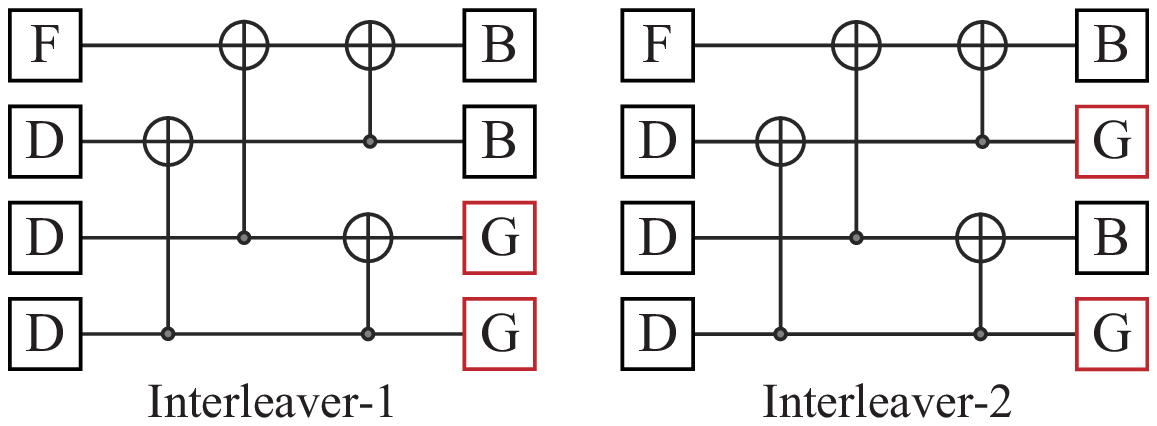}
\caption{\small Polar coding for independent non-identical channels.}
\label{Mapping}
\vspace{-0.3cm}
\end{figure}

\begin{figure}[h]\centering
\includegraphics[width=2.2 in]{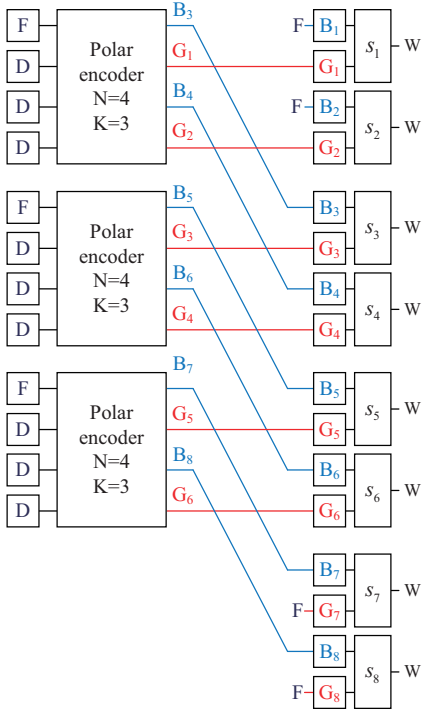}
\caption{\small A toy example of CPCM for 4-PAM.}
\label{cpcm4pam}
\vspace{-0.5cm}
\end{figure}

In Fig. \ref{flowchart}, a message source produces $L$ source words $\mathbf{d}=\{\mathbf{d}^{(1)},\dots,\mathbf{d}^{(L)}\}$, with the $l$-th ($1\leq l\leq L$) source word $\mathbf{d}^{(l)}=\{d_1^{(l)},\dots,d_K^{(l)}\}$ generated uniformly at random over all possible source words of length $K$ over a binary field. These source words are then encoded into $L$ binary polar codes $\mathbf{x}=\{\mathbf{x}^{(1)},\dots,\mathbf{x}^{(L)}\}$, where the $l$-th codeword $\mathbf{x}^{(l)}=\{x_1^{(l)},\dots,x_N^{(l)}\}$ has a block-length $N$, thus inducing a code rate $R=K/N$. Then, these encoded bits are mapped to $\mathbf{v}=\{\mathbf{v}^{(1)},\dots,\mathbf{v}^{(T)}\}$ as the input sequence of the modulator, and its $t$-th ($1\leq t\leq T$) element $\mathbf{v}^{(t)}=\{v_1^{(t)},\dots,v_m^{(t)}\}$ is modulated to a $2^m$-PAM symbol $s_t$. These $T$ symbols $\mathbf{s}=\{s_1,\dots,s_T\}$ are then sent through the channel and $\mathbf{y}=\{y_1,\dots,y_T\}$ is the real number channel output. With these output, the demapper generates the soft information of $\mathbf{x}=\{\mathbf{x}^{(1)},\dots,\mathbf{x}^{(L)}\}$, i.e., $\mathbf{z}=\{\mathbf{z}^{(1)},\dots,\mathbf{z}^{(L)}\}$. Finally, the decoder processes $\mathbf{z}$ and produces an estimate of source words as $\mathbf{\hat{d}}=\{\mathbf{\hat{d}}^{(1)},\dots,\mathbf{\hat{d}}^{(L)}\}$. Specially, in the modulation module, two specific labeling rules in \cite{PCM_Huber}, i.e., set-partition labeling and Gray labeling rules, are considered.

\begin{figure}[t]
\centering
\includegraphics[width=3.3 in]{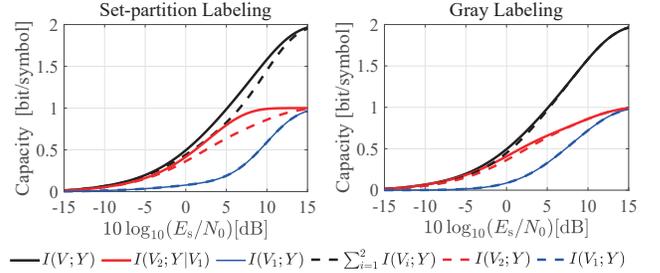}
\caption{\small The constellation capacity of 4-PAM channel and its bit channel capacities using different labeling rules. The solid lines correspond to MLC method, whereas the dashed lines correspond to BICM method.}
\label{4PAM_Bit_level_capa}
\vspace{-0.25cm}
\end{figure}

Although the choice of labeling rules does not effect the constellation capacity of $I(V;Y)$, it does have a strong impact on bit channel capacities. As shown in Fig. \ref{4PAM_Bit_level_capa}, set-partition labeling induces larger mutual information increase from $I(V_1;Y)$ to $I(V_2;Y|V_1)$ compared to Gray labeling, thus is a better choice for schemes that profit from information transfer among bit channels, like MLC. Whereas Gray labeling generates bit channels that are as independent as possible, thus is more suitable for schemes that try to alleviate the dependency created by modulation, like BICM. In CPCM, set-partition labeling is selected as it gives the best performance.

\section{Polar Coding with independent non-identical channels}\label{sec:3}

The authors in \cite{PolarBlockFading_SLiu_icc} have discussed how to map polar coded bits to different channels. They proved that interleaver-$2$ minimizes the sum of even-indexed Bhattacharyya parameters after the first polarization stage, i.e., $\sum_{i=2,4,\dots,N}{Z_1^{(i)}}$, and indicated that optimizing the performance of even-indexed channels in the first stage creates a good starting point for the whole polarization process.

In the following, we investigate the BLER union bound of polar code with interleaver-$1$, interleaver-$2$ and at most $100$ random interleavers over $m$ types of channels.
When $m$ is not of power 2, e.g., $m=3$, polar codes with block-lengths $N=2^n$ become incompatible in the sense that the number of channel uses for each channel type is not an integer. Nevertheless, one can always shorten the current codeword $a_1^N$ into a polar code with block-length $N_s$ that is divisible by $m$. Without loss of generality, we consider half-rate polar codes with $m$ different channels $\{W_1,\ldots,W_m\}$ of a constant average capacity of $\frac{1}{m}\sum_{i=1}^{m}{I(W_i)}=0.7$.

\begin{table*}[htbp]
\caption{\small $20$ Random instances of Bi-AWGN channels $\{W_1, W_2, W_3, W_4| I(W_1)+I(W_2)+I(W_3)+I(W_4)=2.8\}$ }\label{table2}
\vspace{-0.15cm}
\resizebox{\textwidth}{8.5mm}{
\begin{tabular}[h]{|c|c|c|c|c|c|c|c|c|c|c|c|c|c|c|c|c|c|c|c|c|}
\hline
Index    & 1 & 2 & 3 & 4 & 5 & 6 & 7 & 8 & 9 & 10 & 11 & 12 & 13 & 14 & 15 & 16 & 17 & 18 & 19 & 20\\
\hline
$I(W_1)$ &0.65&0.75&0.45&0.65&0.55&0.85&0.25&0.85&0.25&0.15&0.85&0.55&0.75&0.85&0.75&0.45&0.75&0.85&0.85&0.85 \\
\hline
$I(W_2)$ &0.41&0.52&0.78&0.89&0.58&0.56&0.8&0.64&0.7&0.97&0.71&0.39&0.95&0.9&0.83&0.98&0.19&0.8&0.4&0.85\\
\hline
$I(W_3)$ &0.75&0.65&0.65&0.4&0.95&0.8&1&0.85&0.95&0.95&0.7&1&0.15&0.75&0.4&0.4&1&0.25&1&1\\
\hline
$I(W_4)$ &0.99&0.88&0.92&0.86&0.72&0.59&0.75&0.46&0.9&0.73&0.54&0.86&0.95&0.3&0.82&0.97&0.86&0.9&0.55&0.1\\
\hline
\end{tabular}}
\vspace{-0.5cm}
\end{table*}

We start with $m=2$ binary erasure channels (BECs) $\{W_1,W_2\}$, whose Bhattacharyya parameters are $Z_1$ and $Z_2$, respectively.
Fig. \ref{BLER_Union_2BEC} shows the BLER union bound $\sum_{i\in\mathcal{A}}{Z_n^{(i)}}$ for polar codes with different block-lengths. It can be observed that although interleaver-$2$ is not globally optimal for minimizing $\sum_{i\in\mathcal{A}}{Z_n^{(i)}}$, as the block-length gets large, interleaver-$2$ always generates the lowest BLER union bound.

\begin{figure}[t]
\centering
\includegraphics[scale=0.21]{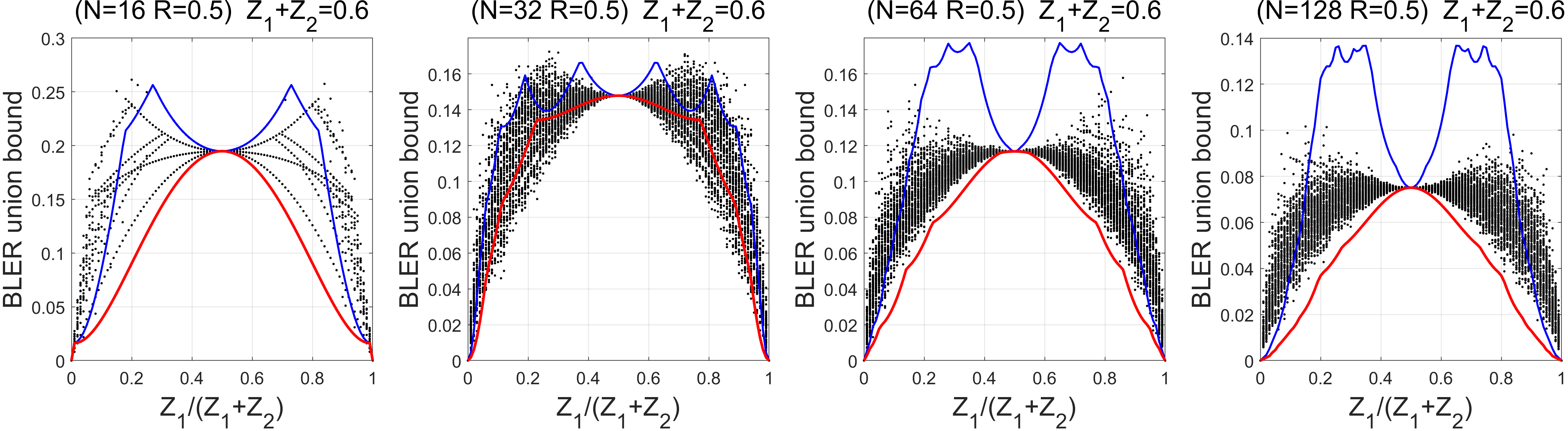}
\caption{ \small The BLER union bound under $2$ types of BECs. Red line corresponds to interleaver-2, blue line corresponds to interleaver-1, and black dots are generated with at most 100 random interleavers.}
\label{BLER_Union_2BEC}
\vspace{-0.5cm}
\end{figure}

For AWGN channels, it is difficult to precisely track the Bhattacharyya parameters of each polarized bit-channel, so we use Gaussian approximation (GA) \cite{GA_Urbanke_2001} to recursively calculate the error probabilities of bit-channels $\{P_{\rm{e}}(W_N^{(i)})|i=1,...,N\}$, thus the BLER union bound can be estimated as $\sum_{i\in\mathcal{A}}{P_{\rm{e}}(W_N^{(i)})}$. For $2$ AWGN channels, similar observation like Fig. \ref{BLER_Union_2BEC} is obtained, hence is omitted here.
For $m=\{3, 4\}$ cases, as the average capacity of these $m$ channels is set to be a constant, once the capacities of $(m-1)$ channels are determined, the remaining one will be determined as well. This enables us to plot the BLER union bound under $3$ different channels in a $3$-dimension plot, as shown in Fig. \ref{Union_bound_GA_3AWGNs}. However, when $m=4$, it is hard for us to arrange all the possible combinations of $\{I(W_1), I(W_2), I(W_3), I(W_4)\}$ in one plot, so we randomly select $20$ instances from all possible values of $\{I(W_i)|i=1,2,3,4\}$ in Table \ref{table2}, and plot their corresponding BLER union bounds in Fig. \ref{Union_bound_GA_4_AWGNs}. It can be observed that the conclusion we drawn from $2$ independent non-identical BECs can be extended to $m=\{3,4\}$ independent non-identical AWGN channels as well, where interleaver-$2$ generally generates the lowest BLER union bounds for polar codes with practical block-lengths (e.g., $N\geq 128$).

\begin{figure}[t]
\centering
\subfigure[$m=3$]{
\begin{minipage}[b]{0.5\textwidth}
\centering
\includegraphics[width=3.4 in]{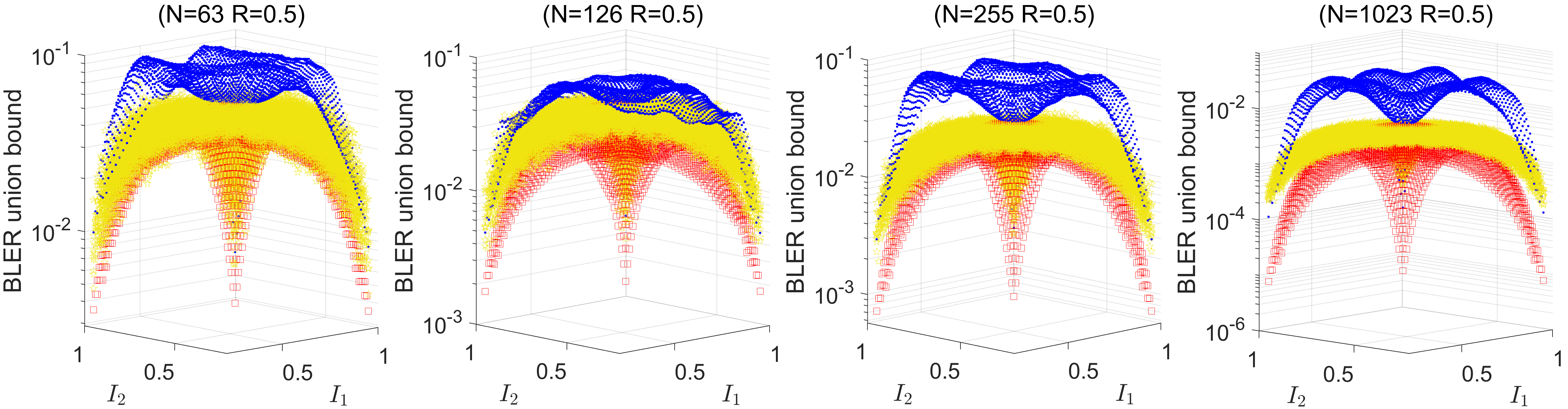}
\label{Union_bound_GA_3AWGNs}
\end{minipage}%
}%
\vspace{-0.25cm}
\centering
\subfigure[$m=4$]{
\begin{minipage}[b]{0.5\textwidth}
\centering
\includegraphics[width=3.4 in]{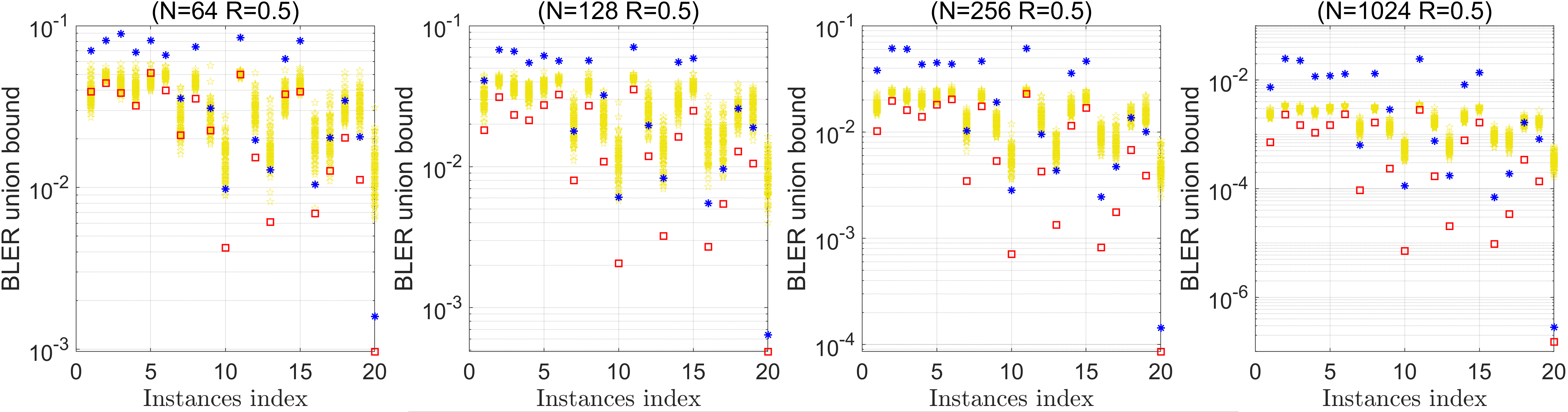}
\label{Union_bound_GA_4_AWGNs}
\end{minipage}%
}%
\vspace{-0.25cm}
\caption{ \small BLER union bound under $m$ types of AWGN channels. Red squares are generated with interleaver-2, blue dots are generated with interleaver-1, and yellow stars are generated with at most 100 random interleavers.}
\vspace{-0.5cm}
\end{figure}

\section{Convolutional Polar Coded Modulation}\label{sec:4}

\subsection{CPCM for 4-PAM}
We start with a toy example in Fig. \ref{cpcm4pam}, where three polar codes $\{\mathbf{x}^{(1)},\mathbf{x}^{(2)},\mathbf{x}^{(3)}\}=\{a_1^N, b_1^N, c_1^N\}$ are transmitted. With MLC principle, each $4$-PAM channel W can be divided into a bad bit channel B with capacity $I(V_1;Y)$ and a good bit channel G with capacity $I(V_2;Y|V_1)$. Also, we specify symbols in $\{s_{(k-1)\frac{N}{2}+1},\dots,s_{k(\frac{N}{2})}\}$ as the $k$-th ($k=\{1,2,3,4\}$) transmission block, thus the channels B and G in the $k$-th transmission block can be explicitly indicated by $v_1^{(t)}$ and $v_2^{(t)}$ with $t=\{(k-1)\frac{N}{2}+1,\dots,k(\frac{N}{2})\}$, respectively.

To transmit $a_1^N$, we map its even-indexed bits to the good channels in the first transmission block, i.e., $\{v_2^{(t)}=a_{i}|i\mod{2}=0;t=1,\dots,\frac{N}{2}\}$, whereas its odd-indexed bits are mapped to the bad channels in the second transmission block by $\{v_1^{(t)}=a_{i}|i\mod{2}=1;t=\frac{N}{2}+1,\dots,N\}$. As for the bad channels in the first transmission block, they are frozen to zeros or ones with equal probability by $\{v_1^{(t)}=\texttt{rand}(0,1)| t=1,\dots,\frac{N}{2}\}$\footnote{One may also set them to some predetermined values, e.g., all-zeros. This is because set-partition labeling rule is used in CPCM, so no matter what value is on a bad bit channel, once it is determined, all the remaining symbols will have an equal Euclidean distance between each other, so the frozen values make no difference in decoding performance.}.
Similarly, to transmit $b_1^N$, we map its even-indexed bits to the good channels in the second transmission block by $\{v_2^{(t)}=b_{i}|i\mod{2}=0;t={\frac{N}{2}+1},\dots,{N}\}$, and its odd-indexed bits are mapped to the bad channels in the third transmission block by $\{v_1^{(t)}=b_{i}| i\mod{2}=1; t=N+1,\dots,\frac{3N}{2}\}$. Such procedure is continued until the last codeword is reached. In this manner, all the polar codes are transmitted over independent bit channels.

Note that for the last codeword in Fig. \ref{cpcm4pam}, i.e., $c_1^N$, instead of freezing the good channels in the last transmission block, better performance can be achieved by freezing the bad channels $B_7$ and $B_8$ while mapping the odd-indexed bit of $c_1^N$ to the good channels $G_7$ and $G_8$. However, this makes $c_1^N$ inconsistent with the former transmitted codewords. To maintain the consistency with previous codewords, all codewords are transmitted in a unified manner in this paper.

For decoding, a bidirectional decoding is adopted. To be more specific, we first start form the first codeword $a_1^N$: once the first transmission block $\{s_{1},\dots,s_{\frac{N}{2}}\}$ is received, the even-indexed LLRs of $a_1^N$ can be calculated directly from $\{v_2^{(t)}|t=1,\dots,\frac{N}{2}\}$ as the values on $\{v_1^{(t)}|t=1,\dots,\frac{N}{2}\}$ are frozen. For its odd-indexed LLRs, they can be obtained from $\{v_1^{(t)}|t=\frac{N}{2}+1,\dots,N\}$ by taking the values on $\{v_2^{(t)}|t=\frac{N}{2}+1,\dots,{N}\}$ as noise. For now, all the LLRs of $a_1^N$ are ready and one can decode $a_1^N$ with successive cancellation (SC) decoding. Once $a_1^N$ is recovered, the values on $\{v_1^{(t)}|t=\frac{N}{2}+1,\dots,{N}\}$ are known, so the even-indexed LLRs of $b_1^N$ can be retrieved directly from the second transmission block. For its odd-indexed LLRs, they can be obtained from $\{v_1^{(t)}|t=N+1,\dots,\frac{3N}{2}\}$ by taking the values on $\{v_2^{(t)}|t=N+1,\dots,\frac{3N}{2}\}$ as noise, thus $b_1^N$ can be recovered as well. Finally, with the decoded values of $\{b_i|i\mod{2}=1\}$ and the frozen bits on $\{v_2^{(t)}| t=\frac{3N}{2}+1,\dots,2N\}$, one can retrieve the even-indexed LLRs and the odd-indexed LLRs of $c_1^N$ from the third transmission block and the forth transmission block, respectively, and $c_1^N$ can be decoded.

To reduce the susceptibility of CPCM to error propagation under large $L$, if an error occurred in the above forward decoding, we record the index of the current erroneous block as $\texttt{Idx\_forward}$ and start decoding from the last codeword. For example, suppose that $b_1^N$ encountered a decoding failure, then we set $\texttt{Idx\_forward}=2$ and start decoding from $c_1^N$. As the values on $\{v_2^{(t)}| t=\frac{3N}{2}+1,\dots,2N\}$ have been frozen, the LLRs of $\{c_i| i \mod2=1\}$ can be calculated directly. For the LLRs of $\{c_i| i \mod2=0\}$, they can be obtained by taking the values on $\{v_1^{(t)}|t=N+1,\dots,\frac{3N}{2}\}$ as noise since $b_1^N$ was wrongly decoded. Now we have the whole LLRs of $c_1^N$, and if $c_1^N$ can be correctly decoded, then the LLRs of the odd-indexed bits in $b_1^N$ can be updated and another attempt of decoding can be activated for $b_1^N$. During such backward decoding process, if a decoding failure happens again, then we record the index of the current erroneous codeword as $\texttt{Idx\_backward}$, and try to decode the codewords from $\texttt{Idx\_forward}$ to $\texttt{Idx\_backward}$ by retrieving the LLRs of each bit from its transmitted symbol while taking the other irrelevant bits in this symbol as noise. This further reduce the effect of error propagation in the sense that, if there is an erroneous codeword $e_1^N$, then its subsequent codeword is decoded without using any information from $e_1^N$. Moreover, if a codeword after $e_1^N$ is correctly recovered with such method, then we can continue the forward decoding from this codeword and upgrading the subsequent bit channels as indicated earlier.
Note that, except the codewords with indexes $\texttt{Idx\_forward}$ and $\texttt{Idx\_backward}$, all other codewords are decoded only once.
An $r$-bit CRC code can be attached after the $K$-bit data in polar code for error detection. However, this degrades the BLER performance with SC decoding as the code rate seen by a SC decoder is actually $(K+r)/N$. In practice, CRC-aided polar code is usually used together with SC list decoders \cite{CA_SCL_KNiu} to assist error correction.

\subsection{CPCM for general $2^m$-PAM}
Now we generalize the above CPCM scheme to $2^m$-PAM constellation with $m\geq 2$. We propose to map a super codeword $\mathbf{x}=\{\mathbf{x}^{(1)},\dots,\mathbf{x}^{(L)}\}$ to a set of $m$-ary tuples $\mathbf{v}=\{\mathbf{v}^{(1)},\dots,\mathbf{v}^{(T)}\}$ with $T=(L+m-1)\frac{N}{m}$, and then modulate these $T$ tuples to $T$ symbols of $2^m$-PAM.

To start, we divide the whole $T$ symbols into $\frac{Tm}{N}$ transmission blocks such that each transmission block has $\frac{N}{m}$ symbols. Then, for the $l$-th codeword $\mathbf{x}^{(l)}$, we map its bits with indexes $\{i|i\mod{m}=q, q=0,...,m-1\}$ to the $k$-th transmission block with $k=l+(m-q)\mod{m}$. To be more specific, the $i$-th ($1 \leq i\leq N$) encoded bit of the $l$-th codeword is mapped to the $j$-th ($1 \leq j\leq m$) element of the $(1+\lfloor{\frac{i-1}{m}}\rfloor)$-th tuple in the $k$-th transmission block, where $i\mod{m}=j\mod{m}$.
In this way, all the bits in a codeword are mapped to different symbols. To summarize, the overall CPCM scheme for $2^m$-PAM is presented in Algorithm \ref{alg:cpcm}.

\begin{algorithm}[!ht]
\LinesNumbered %
\KwIn{$\mathbf{x}=\{x_{1}^{(l)},\dots,x_{N}^{(l)}\},l=\{1,\dots, L\}$}%
\KwOut{$\mathbf{v}=\{v_{1}^{(t)},...,v_{m}^{(t)}\},t=\{1,..., (L+m-1)\frac{N}{m}\}$}%
\For{$t=1;t\leq (L+m-1)\frac{N}{m};t++$}{
	\For{$j=1;j\leq m;j++$}{
		$v_{j}^{(t)}=\texttt{rand}(0,1)$; // Initialization\\
		}
}
\For{$l=1;l\leq L;l++$}{
	\For{$i=1;i\leq N;i++$}{
		\For{$j=1;j\leq m;j++$}{
		$q=(j\mod{m})$;\\
		\If{$q==(i \mod{m})$}{
			$k=l+(m-q)\mod{m}$;\\
			$t=(k-1)\frac{N}{m}+1+\lfloor{\frac{i-1}{m}}\rfloor$;\\
			$v_{j}^{(t)}=x_{i}^{(l)}$;\\				
			}			
		}
	}
}
\caption{CPCM for $2^m$-PAM}
\label{alg:cpcm}
\end{algorithm}

\section{Performance analysis}

For spectral efficiency performance, note that there are some overhead or rate loss in the first and the last transmission block in Fig. \ref{cpcm4pam}. Nevertheless, this can be avoided by making full use of the first transmission block: instead of using an additional transmission block to carry the odd-indexed bits of the last codeword, one can put these bits on the bad channels in the first transmission block. However, this makes the transmission of the codewords inconsistent. In this paper, we admit this overhead and give the following proposition.

\newtheorem{proposition}{\textbf{Proposition}}
\begin{proposition}
\label{proposition1}
CPCM achieves the constellation capacity of $2^m$-PAM as the number of transmitted codewords $L$ and the codeword block-length $N$ both approach infinity.
\end{proposition}
\begin{proof}
According to Algorithm \ref{alg:cpcm}, when transmitting the last codeword $\mathbf{x}^{(L)}$, CPCM maps the bits with indexes $\{i|i\mod{m}=q, q=0,...,m-1\}$ to the $(L+(m-q)\mod{m})$-th transmission block. Thus except the bits with indexes $\{i|i\mod{m}=0\}$ are mapped to the current $L$-th transmission block, all the bits with indexes $\{i|i\mod{m}=q, q=1,...,m-1\}$ are mapped to the transmission blocks with indexes $\{L+1,\dots,L+m-1\}$. On this basis, CPCM needs totally $(L+m-1)$ transmission blocks to transmit $L$ codewords. As each transmission block has $\frac{N}{m}$ symbols, the number of transmitted symbols is $T=(L+m-1)\frac{N}{m}$. With a code rate $R$, each codeword contains $K=NR$ information bits, so the number of transmitted information bits is $LK$, and this leads to an overall coded modulation-rate of
\begin{equation}
R_{\rm{cm}}=\frac{LK}{T}=\frac{mRL}{L+m-1}
\end{equation}

When the block-length of binary polar codes are sufficiently large, the symmetric capacity of binary-input AWGN channel can be achieved in the sense that $\lim\limits_{N\rightarrow\infty}R=1$, then it is straightforward that
$\lim\limits_{N\rightarrow\infty}\lim\limits_{L\rightarrow\infty}{R_{\rm{cm}}=\frac{mRL}{L+m-1}}=m$.
\end{proof}

For complexity analysis, the added complexity of CPCM mainly comes form decoding the codewords with indexes $\texttt{Idx\_forward}$ and $\texttt{Idx\_backward}$, where additional decoding attempts are given. Nevertheless, as more codewords are transmitted, the average decoding complexity for each codeword converges to that of SC decoding.

\section{Numerical results}
In this section, we provide some numerical results and compare the spectral efficiency of CPCM with the existing multi-level polar coding and bit-interleaved polar coded modulation schemes proposed in \cite{PCM_Huber}.

\begin{figure}[!ht]
\centering
  \includegraphics[scale=0.5]{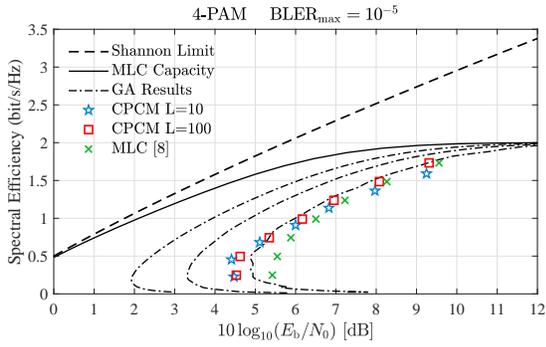}
  \caption{\small Spectral efficiency of CPCM for 4-PAM}
  \label{Spectral_Efficiency_4PAM}
  \vspace{-0.3cm}
\end{figure}

\begin{figure}[!ht]
\centering
  \includegraphics[scale=0.5]{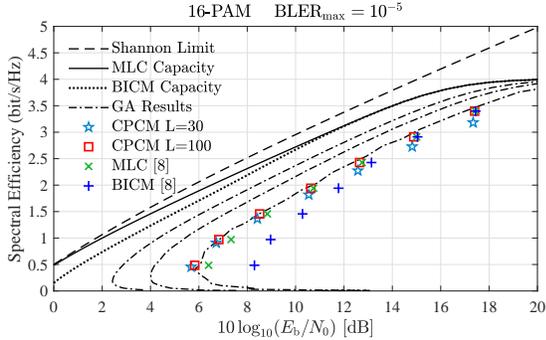}
  \caption{\small Spectral efficiency of CPCM for 16-PAM}
  \label{Spectral_Efficiency_16PAM}
  \vspace{-0.5cm}
\end{figure}

We use interleaver-$2$ and Gray labeling rule for the BICM polar code with block-length $N=512$ as introduced in \cite{PCM_Huber}.
For MLC with $2^m$-PAM, its component polar code has a block-length of $N/m$ so that the overall block-length of MLC is $N=512$, and the code rate of each level is allocated as follows: we first calculate the bit-channel capacities for polar code at each level, then all these bit-channels are concatenated with respect to a nature order from $1$ to $N$, just like the multi-stage decoding order in MLC. Finally, we select the unfrozen bits with higher capacities out of all bit-channels according to the desired overall code rate.
As MLC scheme uses multi-stage decoding, where $m$ polar codes are decoded in a successive manner, the above rate allocation method implicitly allocates a different number of unfrozen bits for each level.

In CPCM, every polar code has a block-length of $N=512$ and shares the same construction as follows: we first calculate the capacities of $m$ bit channels of the $2^m$-PAM constellation with MLC principle as in Fig. \ref{4PAM_Bit_level_capa}, then we map these $m$ types of bit channels to a polar code with interleaver-$2$. Finally, $K=NR$ unfrozen bits are selected by performing density evolution (DE) with GA method \cite{GA_Urbanke_2001}. This procedure is carried out for each value of $E_{\rm{b}}/N_{\rm{0}}$ as in \cite{PCM_Huber}. For ease of comparison, genie-aided SC decoding is adopted to generate the achievable code rate at a target block error rate $\text{BLER}_\text{max}=10^{-5}$ for each marker of ($E_{\rm{b}}/N_{\rm{0}}$, $R_\text{cm}$) in accordance to \cite{PCM_Huber}.
Specially, we transmit $L$ codewords with different code rates of $R=\{\frac{1}{8},\frac{1}{4},\frac{3}{8},\frac{1}{2},\frac{5}{8},\frac{3}{4},\frac{7}{8}\}$ in our simulation.

Fig. \ref{Spectral_Efficiency_4PAM} and Fig. \ref{Spectral_Efficiency_16PAM} show the spectral efficiency performance for $4$-PAM and $16$-PAM, respectively. For each case, we also plot the Shannon bound (under real-valued signals) together with the coded-modulation capacity.
Moreover, the achievable rate of each polar code in CPCM is also plotted for $N=2^n$ (from right to left $n=\{9,12,15\}$) by using GA method, where the curves return back to high $E_{\rm{b}}/N_{\rm{0}}$ values at low rates, as is typical for finite block-length.
For simulated values with small $L$, e.g., $L=10$ and $L=30$ for $4$-PAM and $16$-PAM, respectively, the rate loss of CPCM for both cases is around $9\%$, and this degrades the performance when code rates are relatively high. Nevertheless, when $L$ is large, e.g., $L=100$, the rate loss of CPCM is less than $1\%$ and $3\%$ for $4$-PAM and $16$-PAM, respectively, and it can be observed that CPCM requires less power to achieve a certain rate compared to multi-level polar coding for $4$-PAM, and successfully competes with multi-level polar coding while achieving a significant gain over bit-interleaved polar coded modulation for $16$-PAM.
One may consider the $L$ polar codes in CPCM as a super codeword with block-length $LN$. However, as all the polar codes used in our scheme share one unified construction with block-length $N$, CPCM can be carried out with a encoder and a decoder of size $N$. In this sense, CPCM is more suitable for cases with limited hardware resources.

In summary, this simulation study confirms that CPCM can achieve better spectral efficiency compared to the existing polar coded modulation methods under a fixed number of hardware resources, and the performance improvement is especially pronounced at low $E_{\rm{b}}/N_{\rm{0}}$ values where the coded modulation-rates are relatively low.

\section{Discussion}\label{sec:5}
CPCM is expected to achieved higher achievable rate than its MLC counterpart since a longer (by a factor $m$) codeword can be implemented. Our results are consistent with this expectation.
Compared with BICM, CPCM is more rigorous as the underlying channels used by a polar code are exactly independent. Although one-dimensional constellations over real number AWGN channel are considered in this paper, the extensions to their two-dimensional counterparts, i.e., $M^2$-QAMs over complex channels, are straightforward.

\section*{Acknowledgement}
Helpful discussions with Professor Erdal Ar\*{\i}kan are gratefully acknowledged.

\bibliographystyle{IEEE}
\begin{footnotesize}

\end{footnotesize}

\end{document}